\documentclass[11pt,letterpaper]{article} % from 12pt 5.2.20

\usepackage{amsmath,amsfonts,amsthm,amssymb,stmaryrd,relsize}
\theoremstyle{plain}

\numberwithin{equation}{section}

\newtheorem{thm}{Theorem}[section]
\newtheorem{lem}[thm]{Lemma}

\newenvironment{exam}[1]%  be sure to add \hfill\qedsymbol before \end{exam}
{\begin{flushleft}\textbf{Example #1}.\enspace}%
{\end{flushleft}}

\allowdisplaybreaks  % introduced 7.15.15

\newcounter{cond}

\newcommand{\complex}{{\mathbb C}}
\newcommand{\real}{{\mathbb R}}

\newcommand{\underlambda}{\underline{\lambda}}
\newcommand{\overc}{\overline{c}}

\newcommand{\rmtr}{\mathrm{tr\,}}
\newcommand{\rmdiag}{\mathop{\rm diag}}

\newcommand{\lscript}{\mathcal{L}}
\newcommand{\pscript}{\mathcal{P}}

\newcommand{\ab}[1]{\left|#1\right|}
\newcommand{\doubleab}[1]{\left|\left|#1\right|\right|}
\newcommand{\brac}[1]{\left\{#1\right\}}
\newcommand{\paren}[1]{\left(#1\right)}
\newcommand{\sqbrac}[1]{\left[#1\right]}
\newcommand{\elbows}[1]{{\left\langle#1\right\rangle}}

\begin{document}

\title{SPOOKY ACTION AT A DISTANCE}
\author{Stan Gudder\\ Department of Mathematics\\
University of Denver\\ Denver, Colorado 80208\\
sgudder@du.edu}
\date{}
\maketitle

\begin{abstract}
This article studies quantum mechanical entanglement. We begin by illustrating why entanglement implies action at a distance. We then introduce a simple criterion for determining when a pure quantum state is entangled. Finally, we present a measure for the amount of entanglement for a pure state.
\end{abstract}

\section{Quantum Mechanics in a Nutshell}  % Section 1
Entanglement is an important concept in quantum theory and many scientists believe it is responsible for much of the weirdness and nonintuitive nature of this theory \cite{hhhh09,lh00,nc00}. Albert Einstein called entanglement ``spooky action at a distance'' and there are many people who agree. As strange as it may be, entanglement is a useful resource and it is the underlying basis for the speed and power of quantum computation \cite{gud03,nc00,wc98}.

We begin with a nutshell summary of quantum mechanics. The basic framework consists of a complex Hilbert space $H$ with inner product $\elbows{\phi,\psi}$ and the set of (bounded) linear operators $\lscript (H)$ on $H$. For simplicity, we shall assume that $H$ is finite dimensional in which case $\lscript (H)$ is represented by a set of complex matrices. This is general enough to include the theory of quantum computation and information \cite{gud03,hz12,nc00,wc98}. It also has the advantage of making this article accessible to anyone who has had a first course in linear algebra. In the sequel, when we discuss a quantum state, we shall mean a pure state. There is a more general concept of mixed states, but we only mention these briefly.

A simplified version of the main axioms of quantum mechanics are the following \cite{gud03,hz12,nc00}.

\begin{list} {(\arabic{cond})}{\usecounter{cond}
\setlength{\rightmargin}{\leftmargin}}
%(1)
\item The states of a quantum system are represented by unit vectors in $H$.
%(2)
\item The quantum events are represented by projections on $H$.
%(3)
\item If $\psi\in H$ is a state and $P\in\lscript (H)$ is a projection, then
\begin{equation*}
\pscript _\psi (P)=\elbows{\psi ,P\psi}
\end{equation*}
is the probability that $P$ occurs in the state $\psi$ and if $P$ does occur, then $\psi$ is updated to the state
\begin{equation*}
\psi '=P\psi/\doubleab{P\psi}
\end{equation*}
%(4)
\item If $H_1,H_2$ represent two interacting quantum systems, then the combined system is represented by the tensor product $H_1\otimes H_2$.
%(5)
\item If $P_1,P_2$ are events in system~1 and 2, respectively, then $P_1\otimes I_2$, $I_1\otimes P_2$ are the corresponding events in the combined system, where $I_1\in\lscript (H_1)$, $I_2\in\lscript (H_2)$ are the identity operators.
\end{list}

We now briefly elaborate on these axioms. Quantum mechanics can be thought of as a generalized probability theory \cite{gud03,hz12,lh00}. A unit vector $\psi\in H$ satisfies
\begin{equation*}
\doubleab{\psi}^2=\elbows{\psi ,\psi}=1
\end{equation*}
and $\psi$ gives a quantum probability measure in accordance with (3). Recall that $P\in\lscript (H)$ is a \textit{projection} if
$P^2=P=P^*$. Axiom~(3) connects the abstract concepts of states and events to the outcomes of experiments in the laboratory. Quantum mechanics cannot make precise predictions, it can only produce probabilities for the occurrence of events. The zero operator~$0$ represents the event that never occurs and $I$ represents the event that always occurs. An event $P$ occurs if and only if its \textit{complement} $P'=I-P$ does not occur and
\begin{equation*}
\pscript _\psi (P')=\elbows{\psi ,(I-P)\psi}=\elbows{\psi,\psi}-\elbows{\psi ,P\psi}=1-\pscript _\psi (P)
\end{equation*}

The updating $\psi\mapsto\psi '$ in Axiom~(3) is sometimes called the ``collapse'' of the state upon performing a measurement. This corresponds to the ``ontic'' viewpoint in which $\psi$ is considered to be a real, physical object \cite{lh00}. Another way of viewing this is that once the occurrence of $P$ is confirmed, this information gives a ``more precise'' state $\psi '$. We then call
$\psi '$ the state $\psi$ \textit{conditioned on the occurrence} of the event $P$, which is similar to a conditional probability of ordinary statistics. This corresponds to the ``epistemic'' viewpoint in which $\psi$ is not considered to be a physical object but is only a carrier of our knowledge of the system \cite{lh00}. One can employ either of these philosophies and still get the same quantum predictions. Once we're in the state $\psi '$, then $P$ must occur because Axiom~(3) gives
\begin{equation*}
\pscript_{\psi '}(P)=\elbows{\psi ',P\psi '}=\frac{1}{\doubleab{P\psi}^2}\,\elbows{P\psi ,P^2\psi}
  =\frac{1}{\doubleab{P\psi}^2}\,\elbows{P\psi ,P\psi}=1
\end{equation*}
Notice that there is no problem with dividing by $\doubleab{P\psi}$ because if $P$ has occurred, then
\begin{equation*}
\doubleab{P\psi}^2=\elbows{P\psi ,P\psi}=\elbows{\psi ,P\psi}=\pscript _\psi (P)\ne 0
\end{equation*}
Also notice that since $P$ or $P'$ must occur, the state $\psi$ is updated to $P\psi/\doubleab{P\psi}$ or
$P'\psi/\doubleab{P'\psi}$ when $P$ is tested.

The tensor product $H_1\otimes H_2$ is an important way to combine two Hilbert spaces. If $\brac{\phi _i}$, $\brac{\psi _j}$ are orthonormal bases for $H_1,H_2$, respectively, then by definition, any vector in $H_1\otimes H_2$ has the form
\begin{equation}                % equation (1.1)
\label{eq11}
\gamma =\sum _{i,j}c_{ij}\phi _i\otimes\psi _j,\quad c_{ij}\in\complex
\end{equation}
The main properties of $H_1\otimes H_2$ are that $\phi\otimes\psi$ is linear in both arguments and that
\begin{equation*}
\elbows{\alpha _1\otimes\beta _1,\alpha _2\otimes\beta _2}=\elbows{\alpha _1,\alpha _2}\elbows{\beta _1,\beta _2}
\end{equation*}
It follows that $\gamma$ in \eqref{eq11} is a state if and only if $\sum\ab{c_{ij}}^2=1$. If $A\in\lscript (H_1)$, $B\in\lscript (H_2)$ then $A\otimes B\in\lscript (H_1\otimes H_2)$ and
\begin{equation*}
A\otimes B(\alpha\otimes\beta )=A\alpha\otimes B\beta
\end{equation*}
and any operator on $H_1\otimes H_2$ has the form $\sum c_{ij}A_i\otimes B_j$, $c_{ij}\in\complex$. If $P_1$ is an event in system~1, then $P_1\otimes I_2$ is the corresponding event in the combined system because $P_1\otimes I_2$ occurs in $H_1\otimes H_2$ if and only if $P_1$ occurs in $H_1$. Another way of describing Axiom~(5) is that if $P_1$ is an event in system~1 and we test whether $P_1$ occurs in the combined system, then this test should not be affected by system~2. This statement is made precise in the next lemma where $\overc$ denotes the complex conjugate of $c\in\complex$. Also, note that we use the physics convention that the inner product is anti-linear in the first argument.

\begin{lem}    % Lemma 1.1
\label{lem11}
If $\gamma\in H_1\otimes H_2$ is a state and $P_1\in\lscript (H_1)$ is an event, then there exist states $\alpha _i\in H_1$ and
$\lambda _i\in\sqbrac{0,1}$ with $\sum\lambda _i=1$ such that
\begin{equation}                % equation (1.2)
\label{eq12}
\elbows{\gamma ,P_1\otimes I_2\gamma}=\sum\lambda _i\elbows{\alpha _i,P_1\alpha _i}
\end{equation}
\end{lem}
\begin{proof}
The state $\gamma$ has the form \eqref{eq11} with $\sum\ab{c_{ij}}^2=1$. We then have that
\begin{align*}
\elbows{\gamma ,P_1\otimes I_2\gamma}
  &=\elbows{\sum c_{ij}\phi _i\otimes\psi _j,P_1\otimes I_2\sum c_{rs}\phi _r\otimes\psi _s}\\
  &=\sum\overc _{ij}\sum c_{rs}\elbows{\phi _i\otimes\psi _j,P_1\phi _r\otimes\psi _s}\\
  &=\sum _{i,j,r}\overc _{ij}c_{rj}\elbows{\phi _i,P_1\psi _r}
\end{align*}
Letting
\begin{equation*}
\alpha _j=\sum _ic_{ij}\phi _i/\sum _i\ab{c_{ij}}^2
\end{equation*}
and $\lambda _j=\sum _i\ab{c_{ij}}^2$ we conclude that $\alpha _j$ are states in $H_1$, $\sum\lambda _j=1$ and that
\eqref{eq12} holds.
\end{proof}

Equation~\eqref{eq12} says that there are states $\alpha _i\in H_1$ such that $\pscript _\gamma (P_1\otimes I_2)$ is a convex combination of $\pscript _{\alpha _i}(P_1)$. In this way, a test of the occurrence of $P_1\otimes I_2$ only depends on states in
$H_1$ and is independent of system~2. The function $P_1\mapsto\sum\lambda _i\elbows{\alpha _i,P_1\alpha _i}$ in \eqref{eq12} is called a \textit{mixed state} \cite{hz12,nc00} and we shall not pursue these further.

We now come to our main definition. A state $\gamma\in H_1\otimes H_2$ is \textit{factorized} (or a \textit{product state}) if
$\gamma =\alpha\otimes\beta$ for states $\alpha\in H_1$, $\beta\in H_2$. Otherwise, $\gamma$ is \textit{entangled}. When
$\gamma =\alpha\otimes\beta$, we call $\alpha$ and $\beta$ the \textit{local parts} of $\gamma$. The local parts are not unique because if $\alpha$ and $\beta$ are local parts, then so are $e^{i\theta}\alpha$, $e^{-i\theta}\beta$ where $\theta\in\real$,
$i=\sqrt{-1\,}$.

\begin{exam}{1}  % Example 1
If $\alpha ,\beta$ are orthogonal states in $H$, then
\begin{equation}                % equation (1.3)
\label{eq13}
\delta=\tfrac{1}{\sqrt{2\,}}\,(\alpha\otimes\beta -\beta\otimes\alpha )
\end{equation}
is an entangled state in $H\otimes H$. To show this suppose that $\delta =\phi\otimes\psi$ for states $\phi ,\psi\in H$. Taking the inner product with $\alpha\otimes\alpha$ gives $\elbows{\alpha ,\phi}\elbows{\alpha ,\psi}=0$ so $\elbows{\alpha ,\phi}=0$ or 
$\elbows{\alpha ,\psi}=0$. Suppose $\elbows{\alpha ,\phi}=0$ and take the inner product with $\alpha\otimes\beta$ to obtain
\begin{equation*}
\tfrac{1}{\sqrt{2\,}}=\elbows{\alpha ,\phi}\elbows{\beta ,\psi}
\end{equation*}
which gives a contradiction. If $\elbows{\alpha ,\psi}=0$, take the inner product with $\beta\otimes\alpha$ to again get a contradiction. Hence $\delta$ is entangled.\hfill\qedsymbol
\end{exam}

\section{Action at a Distance}  % Section 2
Alice and Bob prepare an interacting pair of electrons in a state $\gamma$ at a lab in New York; Alice keeps her electron (system~1) in New York and Bob takes his electron (system~2) to a lab on the moon. Since all electrons have the same properties, both systems give a copy of a Hilbert space $H$ and the combined system has Hilbert space $H\otimes H$. For this experiment, $\gamma$ is the factorized state $\alpha\otimes\beta$ where $\alpha ,\beta$ are orthogonal states in $H$. Let $P$ be an event that pertains to an electron (say, spin-up in the $z$-direction). Alice sends her electron through an apparatus that tests $P$. She confirms that $P$ occurs so by Axioms~(3) and (5), $\gamma$ updates to
\begin{equation*}
\gamma '=\frac{P\otimes I(\alpha\otimes\beta )}{\doubleab{P\otimes I(\alpha\otimes\beta )}}
  =\frac{P\alpha\otimes\beta}{\doubleab{P\alpha}}
\end{equation*}
It is fairly clear that Bob's electron is unaffected. To make sure, suppose Bob tests an event $Q$ pertaining to his electron. According to $\gamma$, we have
\begin{equation*}
\pscript _\gamma (I\otimes Q)=\elbows{\alpha\otimes\beta ,I\otimes Q(\alpha\otimes\beta )}
  =\elbows{\alpha\otimes\beta ,\alpha\otimes Q\beta}=\elbows{\beta ,Q\beta}
\end{equation*}
Moreover, we obtain
\begin{align*}
\pscript _{\gamma '}(I\otimes Q)&=\frac{1}{\doubleab{P\alpha}^2}\elbows{P\alpha\otimes\beta ,(I\otimes Q)(P\alpha\otimes\beta )}\\
  \noalign{\medskip}
  &=\frac{1}{\doubleab{P\alpha}^2}\elbows{P\alpha\otimes\beta ,P\alpha\otimes Q\beta}=\elbows{\beta ,Q\beta}
\end{align*}
We conclude that Bob's electron is in the same state after Alice's measurement as it was before.

Next Bob returns to New York where he and Alice prepare a pair of electrons in the entangled state $\delta$ of \eqref{eq13}. Alice keeps her electron in New York and Bob again takes his to the moon. Now Alice performs her experiment that tests $P$ and confirms that $P$ occurs. (It actually doesn't matter if $P$ occurs. If it doesn't, then $P'$ occurs and our conclusion will be the same.) The updated state becomes
\begin{equation*}
\delta '=\frac{(P\otimes I)\delta}{\doubleab{(P\otimes I)\delta}}
\end{equation*}
We have that
\begin{align*}
\doubleab{(P\otimes I)\delta}^2
  &=\tfrac{1}{2}\elbows{P\alpha\otimes\beta -P\beta\otimes\alpha ,P\alpha\otimes\beta -P\beta\otimes\alpha}\\
  &=\tfrac{1}{2}\paren{\elbows{\alpha ,P\alpha}+\elbows{\beta ,P\beta}}
\end{align*}
Letting $N=\tfrac{1}{\sqrt{2\,}}\,\paren{\elbows{\alpha ,P\alpha}+\elbows{\beta ,P\beta}}^{1/2}$ we obtain
\begin{equation*}
\delta '=\tfrac{1}{N}\,(P\otimes I)\delta =\tfrac{1}{N}\,(P\alpha\otimes\beta -P\beta\otimes\alpha )
\end{equation*}
It appears as if the state of Bob's electron is instantaneously changed without Bob doing anything. This is the spooky action at a distance that bothered Einstein and others. This effect in various forms really happens because it has been exhibited in thousands of experiments around the world. It does not violate special relativity which postulates that no signal or object can move faster than the speed of light. This is because further study shows that this action cannot relay any useable information or communication \cite{bus03,hhhh09,lh00}.

In order to be sure about this, let us show that the state of Bob's electron has indeed been changed. As before, let $Q$ be an event for Bob's electron. Before Alice made her measurement, the joint state was $\delta$ and the probability that $Q$ occurs is
\begin{align}                % equation (2.1)
\label{eq21}
\elbows{\delta ,I\otimes Q\delta}&=\tfrac{1}{2}\elbows{\alpha\otimes\beta -\beta\otimes\alpha ,\alpha\otimes Q\beta -\beta\otimes Q\alpha}\notag\\
&=\tfrac{1}{2}\,\paren{\elbows{\beta ,Q\beta}+\elbows{\alpha ,Q\alpha}}
\end{align}
After Alice makes her measurement, the joint state is $\delta '$ and the probability that $Q$ occurs becomes
\begin{align}                % equation (2.2)
\label{eq22}
\elbows{\delta ',I\otimes Q\delta '}
  &=\tfrac{1}{N^2}\,\paren{P\alpha\otimes\beta -P\beta\otimes\alpha ,P\alpha\otimes Q\beta -P\beta\otimes Q\alpha}\notag\\
  &=\tfrac{1}{N^2}\,\left(\elbows{\alpha ,P\alpha}\elbows{\beta ,Q\beta}-\elbows{\alpha ,P\beta}\elbows{\beta ,Q\alpha}\right.\\
  &\qquad\left. -\elbows{\beta ,P\alpha}\elbows{\alpha ,Q\beta}+\elbows{\beta ,P\beta}\elbows{\alpha ,Q\alpha}\right)\notag
\end{align}
We see that \eqref{eq21} and \eqref{eq22} are definitely different. As a simple example, suppose $P=Q=P_\alpha$ is the
one-dimensional projection given by $P_\alpha\psi =\elbows{\alpha ,\psi}\alpha$ for all $\psi\in H$. Then \eqref{eq21} is $1/2$ while \eqref{eq22} is zero. We conclude that Alice's measurement on her electron has instantaneously altered Bob's electron.

\section{Entangled or Not Entangled}  % Section 3
This section presents a simple criterion that determines whether a state is entangled or not \cite{gud191}.

\begin{exam}{2}  % Example 2
It is not so easy to tell whether a state is entangled. Let $\brac{\phi _i}$, $\brac{\psi _i}$, $i=1,2,3$, be orthonormal bases for $H_1,H_2$, respectively and form the joint state $\psi\in H_1\otimes H_2$ given by
\begin{align*}
\psi&=\tfrac{1}{N}\,(4\phi _1\otimes\psi _1-3i\phi _1\otimes\psi _2+5\phi _1\otimes\psi _3
  -8\phi _2\otimes\psi _1+6i\phi _2\otimes\psi _2\\
  &\qquad -10\phi _3\otimes\psi _3+12\phi _3\otimes\psi _1-9i\phi _3\otimes\psi _2+15\phi _3\otimes\psi _3)
\end{align*}
where $i=\sqrt{-1\,}$ and $N=10\sqrt{7\,}$ is the norm of the vector in parentheses. Is $\psi$ factorized and if it is, what are the local parts of $\psi$? If you can answer this outright, you're better than I. I prefer to use the following result \cite{gud191,gud192}.
\hfill\qedsymbol
\end{exam}

\begin{thm}    % Theorem 3.1
\label{thm31}
Let $\brac{\phi _1,\phi _2,\ldots ,\phi _m},\brac{\psi _1,\psi _2,\ldots ,\psi _n}$ be orthonormal bases for $H_1,H_2$, respectively and let
\begin{equation*}
\psi =\sum c_{ij}\phi _i\otimes\psi _j\in H_1\otimes H_2
\end{equation*}
be a state, where $\sum c_{ij}\ne 0$. Then $\psi$ is factorized if and only if for all $i=1,2,\ldots ,m$, $j=1,2,\ldots ,n$ we have that
\begin{equation}                % equation (3.1)
\label{eq31}
c_{ij}\sum _{i,j}c_{ij}=\sum _jc_{ij}\sum _ic_{ij}
\end{equation}
Moreover, if $\psi$ is factorized, then the local parts are $\alpha /\doubleab{\alpha}$, $\beta /\doubleab{\beta}$ where
$\alpha =\sum a_i\phi _i$, $\beta =\sum b_j\psi _j$, $a_i=\tfrac{1}{c}\,\sum _jc_{ij}$, $b_j=\sum _ic_{ij}$, $c=\sum _{ij}c_{ij}$.
\end{thm}
\begin{proof}
We know that $\psi$ is factorized if and only if $\psi =\alpha\otimes\beta$ for some states $\alpha\in H_1$, $\beta\in H_2$. Let
$\alpha =\sum a_i\phi _i$ and $\beta =\sum b_j\psi _j$. It follows that
\begin{equation*}
\sum _{i,j}c_{ij}\phi _i\otimes\psi _j=\paren{\sum _ia_i\phi _i}\otimes\paren{\sum _jb_j\psi _j}=\sum _{i,j}a_ib_j\phi _i\otimes\psi _j
\end{equation*}
Hence, $\psi$ is factorized if and only if there exist sequences of complex numbers $\brac{a_i}$, $\brac{b_j}$, $i=1,2,\ldots ,m$,
$j=1,2,\ldots ,n$ such that $c_{ij}=a_ib_j$ where $\sum\ab{a_i}^2=\sum\ab{b_j}^2=1$. If \eqref{eq31} holds, letting
$c=\sum _{i,j}c_{ij}$, $a'_i=\tfrac{1}{c}\,\sum _jc_{ij}$, $b'_j=\sum _ic_{ij}$ we have that $c_{ij}=a'_ib'_j$. Since $\doubleab{\psi}=1$ we have that
\begin{equation*}
\sum _i\ab{a'_i}^2\sum _j\ab{b'_j}^2=\sum _{i,j}\ab{c_{ij}^2}=1
\end{equation*}
Letting $a_i=a'_i\Big/\sqrt{\sum\ab{a'_i}^2\,}$ and $b_j=b'_j\Big/\sqrt{\sum\ab{b'_j}^2}$ we obtain $c_{ij}=a_ib_j$ and
$\sum\ab{a_i}^2=\sum\ab{b_j}^2=1$. Hence, $\psi$ is factorized. Conversely, suppose that $\psi$ is factorized so there exist sequences $\brac{a_i}$, $\brac{b_j}$ with $c_{ij}=a_ib_j$. We conclude that
\begin{equation*}
\sum _jc_{ij}\sum _ic_{ij}=a_ib_j\sum _{i,j}a_ib_j=c_{ij}\sum _{i,j}c_{ij}
\end{equation*}
so \eqref{eq31} holds. The last sentence follows from our previous work.
\end{proof}

Theorem~3.1 gives a necessary and sufficient condition for factorizability under the condition that $\sum c_{ij}\ne 0$. A more complicated criterion than \eqref{eq31} gives such a characterization in terms of the $c_{ij}$ without this condition \cite{gud191}.
Theorem~4.1 will give another characterization.

\begin{exam}{3}  % Example 3
We use Theorem~3.1 to answer the question in Example~2. Except for the common factor $1/N$ we have that $c_{11}=4$,
$c_{12}=-3i$, $c_{13}=5$, $c_{21}=-8$, $c_{22}=6i$, $c_{23}=-10$, $c_{31}=12$, $c_{32}=-9i$, $c_{33}=15$. We then obtain
$\sum _{i,j}c_{ij}=6(3-i)$, $\sum _jc_{1j}=3(3-i)$, $\sum _jc_{2j}=6(i-3)$, $\sum _jc_{3j}=9(3-i)$, $\sum _ic_{i1}=8$, 
$\sum _ic_{i2}=-6i$, $\sum _ic_{i3}=0$. It is easy to check that \eqref{eq31} holds (the factor $1/N$ cancels from both sides). Also, the local parts become
\begin{equation*}
\alpha =\frac{1}{\sqrt{14\,}}\,(\phi _1-2\phi _2+3\phi _3),\quad\beta =\frac{1}{5\sqrt{2\,}}\,(4\psi _1-3i\psi _2+5\psi _3)
\rlap{$\qquad\qquad \Box$}
\end{equation*}
\end{exam}

The following Schmidt decomposition theorem \cite{hz12,wc98} is important in this work.

\begin{thm}    % Theorem 3.2
\label{thm32}
Any state $\psi\in H_1\otimes H_2$ has a \textit{Schmidt decomposition}
\begin{equation}                % equation (3.2)
\label{eq32}
\psi =\sum _{i=1}^r\sqrt{\lambda _i\,}\phi _i\otimes\psi_i
\end{equation}
where $\lambda _i>0$, $\sum\lambda _i=1$ and $\brac{\phi _i}$, $\brac{\psi _i}$ are orthonormal vectors in $H_1,H_2$, respectively.
\end{thm}

In the Schmidt decomposition, the \textit{singular-values} $\sqrt{\lambda _i\,}$ are unique, $i=1,2,\ldots ,r$. It follows that $\psi$ is factorized if and only if $r=1$ in \eqref{eq32}. Why not just use this to test whether $\psi$ is factorized? One reason is that the Schmidt decomposition can be difficult to construct. Another reason is that Theorem~3.1 generalizes to multipartite systems (more than two parts) where no Schmidt decomposition is available \cite{gud191}.

\section{An Entanglement Measure}  % Section 4
We now present a measure of entanglement. Using this measure, we can decide how entangled a state is and when one state is more entangled than another. If a state $\psi\in H_1\otimes H_2$ has Schmidt decomposition
$\psi =\sum _{i=1}^r\sqrt{\lambda _i\,}\,\phi _i\otimes\psi _i$, $\lambda _i>0$, $\sum\lambda _i=1$, then
$\underlambda =(\lambda _1,\lambda _2,\ldots ,\lambda _r)$ is a probability distribution. The \textit{entanglement number} of $\psi$ is
\begin{equation}                % equation (4.1)
\label{eq41}
e(\psi )=\sqbrac{1-\sum _{i=1}^r\lambda _i^2}^{1/2}=\sqbrac{\sum_{i\ne j}\lambda _i\lambda _j}^{1/2}
  =\sqbrac{\sum _i\lambda _i(1-\lambda _i)}^{1/2}
\end{equation}
Relative to the distribution $\underlambda$, the last expression in \eqref{eq41} shows that $e(\psi )$ is the average deviation of
$\underlambda$ from $1$. Notice that $\psi$ is factorized if and only if $e(\psi )=0$ which is an important property for an entanglement measure.

There are various justifications for the definition \eqref{eq41} \cite{gud192}. One is that if the distribution $\underlambda$ is peaked near $1$, then $e(\psi )$ should be near $0$ and if $\lambda$ is spread fairly equally, then $e(\psi )$ should be large. For example, suppose $\psi$ has Schmidt decomposition
\begin{equation*}
\psi =\frac{\sqrt{99\,}}{10}\,\phi _1\otimes\psi _1+\frac{1}{10}\,\phi _2\otimes\psi _2
\end{equation*}
Then $\underlambda =\paren{\tfrac{99}{100},\tfrac{1}{100}}$ and $\psi$ has the dominate factorized term
$\tfrac{\sqrt{99\,}}{100}\,\phi _1\otimes\psi _1$ together with a very subordinate term $\tfrac{1}{10}\,\phi _2\otimes\psi _2$ so $e(\psi )$ should be small. Indeed,
\begin{equation*}
e(\psi )=\sqbrac{1-\paren{\frac{99}{100}}^2-\paren{\frac{1}{100}}^2}^{1/2}\approx 0.14
\end{equation*}

We call $r$ in the Schmidt decomposition of $\psi$ the \textit{index} of $\psi$ and write $n(\psi )=r$. We say that $\psi$ is
\textit{maximally entangled with index} $n(\psi )=r\ge 2$, if $\lambda _i=1/r$, $i=1,2,\dots ,r$. In this case, the distribution
$\underlambda$ is uniformly spread so the entanglement should be large. The next result verifies this and is a standard calculus maximization problem whose proof we leave to the reader.

\begin{thm}    % Theorem 4.1
\label{thm41}
$e(\psi )\le\sqbrac{\tfrac{n(\psi )-1}{n(\psi )}}^{1/2}$ and equality is achieved if and only if $\psi$ is maximally entangled with index $n(\psi )$.
\end{thm}

\begin{exam}{4}  % Example 4
Let $\brac{\phi _1,\phi _2,\phi _3},\brac{\psi _1,\psi _2,\psi _3}$ be orthonormal bases for $H_1,H_2$, respectively. Define the following states:
\begin{align*}
\alpha&=\tfrac{1}{\sqrt{2\,}}\,\phi _1\otimes\psi _1+\tfrac{1}{\sqrt{2\,}}\,\phi _2\otimes\psi _2\\
\beta&=\tfrac{1}{\sqrt{3\,}}\,\phi _1\otimes\psi _1+\tfrac{1}{\sqrt{3\,}}\,\phi _2\otimes\psi _2+\tfrac{1}{\sqrt{3\,}}\,\phi _3\otimes\psi _3\\
\gamma&=\tfrac{1}{\sqrt{2\,}}\,\phi _1\otimes\psi _1+\tfrac{1}{\sqrt{3\,}}\,\phi _2\otimes\psi _2
    +\tfrac{1}{\sqrt{6\,}}\,\phi _3\otimes\psi _3\\
\delta&=\tfrac{1}{3}\,\phi _1\otimes\psi _1+\tfrac{1}{3}\,\phi _2\otimes\psi _2+\sqrt{\tfrac{7}{9}\,}\,\phi _3\otimes\psi _3
\end{align*}
The distributions for these states are $\paren{\tfrac{1}{2},\tfrac{1}{2}}$, $\paren{\tfrac{1}{3},\tfrac{1}{3},\tfrac{1}{3}}$,
$\paren{\tfrac{1}{2},\tfrac{1}{3},\tfrac{1}{6}}$, $\paren{\tfrac{1}{9},\tfrac{1}{9},\tfrac{7}{9}}$, respectively. We see that $\alpha ,\beta$ are maximally entangled with indexes $2,3$, respectively and $e(\alpha )=\tfrac{1}{\sqrt{2\,}}$, $e(\beta )=\sqrt{\tfrac{2}{3}\,}$,
$e(\gamma )=\sqrt{\tfrac{11}{18}\,}$, $e(\delta )=\sqrt{\tfrac{30}{9}\,}$. Hence,
\begin{equation*}
e(\delta )<e(\alpha )<e(\gamma )<e(\beta )\rlap{$\qquad\qquad \Box$}
\end{equation*}
\end{exam}

We mentioned in Section~3 that the Schmidt decomposition can be hard to compute. This is especially true for large index $r$ and is similar to finding the eigenvalues for an $r\times r$ matrix. We conclude that finding $e(\psi )$ using \eqref{eq41} can be quite difficult. We now give an efficient method for finding $e(\psi )$ that applies to any orthonormal basis $\brac{\phi _i},\brac{\psi  _j}$ for $H_1,H_2$, respectively. Let $\psi =\sum c_{ij}\phi _i\otimes\psi _j$ and define the matrix $C=\sqbrac{c_{ij}}$. Denoting the adjoint of $C$ by $C^*$, we define the positive semidefinite, square matrix $\ab{C}=(C^*C)^{1/2}$. For a square matrix $A=\sqbrac{a_{ij}}$, we define the \textit{trace} of $A$ by $\rmtr (A)=\sum \limits_ia_{ii}$.

\begin{thm}    % Theorem 4.2
\label{thm42}
{\rm{(a)}}\enspace $e(\psi )=\sqbrac{1-\rmtr\paren{\ab{C}^4}}^{1/2}$.
{\rm{(b)}}\enspace $\psi$ is factorized if and only if $\rmtr\paren{\ab{C}^4}=1$.
{\rm{(c)}}\enspace We have that
\begin{equation*}
\rmtr\paren{\ab{C}^4}=\sum _{r,s}\sqbrac{\Big|\sum _i\overc _{ri}c_{si}\Big|^2}
\end{equation*}
\end{thm}
\begin{proof}
(a)\enspace By the singular-value theorem \cite{nc00}, we can write $C=UDV$ where $U,V$ are unitary matrices and $D$ is a diagonal matrix
\begin{equation*}
D=\rmdiag(\lambda _1^{1/2},\lambda _2^{1/2},\ldots ,\lambda _n^{1/2})
\end{equation*}
with $\lambda _i\ge 0$ and $\sqrt{\lambda _i\,}$ are the singular-values of $C$. These singular-values coincide with those given in \eqref{eq32}. Now
\begin{equation*}
\ab{C}^2=C^*C=V^*DU^*UDV=V^*D^2V
\end{equation*}
and hence, $\ab{C}^4=V^*D^4V$. Therefore
\begin{equation*}
\sum\lambda _i^2=\rmtr (D^4)=\rmtr\paren{\ab{C}^4}
\end{equation*}
We conclude from \eqref{eq41} that (a) holds.
(b)\enspace follows from (a). To verify (c) we have that 
\begin{align*}
\ab{C}_{ij}^4&=(C^*C)(C^*C)_{ij}=\sum _k(C^*C)_{ik}(C^*C)_{kj}\\
  &=\sum _{k,r}C_{ir}^*C_{rk}\sum C_{ks}^*C_{sj}=\sum _{r,s,k}\overc _{ri}c_{rk}\overc _{sk}c_{sj}
\end{align*}
Hence,
\begin{align*}
\rmtr\paren{\ab{C}^4}&=\sum _i\ab{C}_{ii}^4=\sum _{i,k,r,s}\overc _{ri}c_{rk}\overc _{sk}c_{si}\\
 &=\sum _{r,s}\sqbrac{\sum _i\overc _{ri}c_{si}\sum _k\overline{\overc _{rk}c_{sk}}}
   =\sum _{r,s}\sqbrac{\Big|\sum _i\overc_{ri}c_{si}\Big|^2}\qedhere
\end{align*}
\end{proof}

\begin{exam}{5}  % Example 5
Let $H$ be a 2-dimensional Hilbert space with orthonormal basis $\brac{\phi _1,\phi _2}$ and let $\psi\in H\otimes H$ be the state
\begin{equation*}
\psi =\tfrac{1}{\sqrt{10}\,}\,(\phi _1\otimes\phi _1-2i\phi _1\otimes\phi _2+\phi _2\otimes\phi _1-2i\phi _2\otimes\phi _2)
\end{equation*}
We then have that
\begin{equation*}
C=\frac{1}{\sqrt{10\,}}\begin{bmatrix}1&-2i\\1&-2i\\\end{bmatrix}
\end{equation*}
Applying Theorem~4.2(c) gives
\begin{align*}
\rmtr\paren{\ab{C}^4}&=\sum _{r,s=1}^2\ab{\overc _{r1}c_{s1}+\overc _{r2}c_{s2}}^2\\
  &=\paren{\ab{c_{11}}^2+\ab{c_{12}}^2}+\paren{\ab{c_{21}}^2+\ab{c_{22}}^2}+2\ab{\overc _{11}c_{21}+\overc _{21}c_{22}}^2\\
  &=\frac{1}{100}\sqbrac{(1+4)^2+(1+4)^2+2(1+4)^2}=1
\end{align*}
It follows from Theorem~4.2(b) that $\psi$ is factorized. In fact, $\psi =\alpha\otimes\beta$ where
$\alpha =\tfrac{1}{\sqrt{2\,}}\,(\phi _1+\phi _2)$, $\beta =\tfrac{1}{\sqrt{5\,}}\,(\phi _1-2i\phi _2)$.\hfill\qedsymbol
\end{exam}

\begin{exam}{6}  % Example 6
We change Example~5 slightly and let
\begin{equation*}
\psi =\tfrac{1}{\sqrt{10\,}}\,(\phi _1\otimes\phi _1=2i\phi _1\otimes\phi _2+\phi _2\otimes\phi _1+2i\phi _2\otimes\phi _2)
\end{equation*}
We then have 
\begin{equation*}
C=\frac{1}{\sqrt{10\,}}\begin{bmatrix}1&-2i\\1&2i\\\end{bmatrix}
\end{equation*}
As in Example~5
\begin{equation*}
\rmtr\paren{\ab{C}^4}=\tfrac{1}{100}\sqbrac{(1+4)^2+(1+4)2+2(1-4)^2}=\tfrac{17}{25}
\end{equation*}
We conclude that $\psi$ is entangled with
\begin{equation*}
e(\psi )=\paren{1-\frac{17}{25}}^{1/2}=\frac{2\sqrt{2\,}}{5}\rlap{$\qquad\qquad \Box$}
\end{equation*}
\end{exam}

 In the previous two examples, we used Theorem~4.2(c) to find $e(\psi )$. However, it is usually easier just to find $\ab{C}^4$ directly and use Theorem~4.1(a) as the following example shows.

\begin{exam}{7}  % Example 7
For $H$ defined as in Examples~5 and~6, let
\begin{equation*}
\psi =\tfrac{\sqrt{3\,}}{2}\,\phi _1\otimes\phi _1+\tfrac{1}{2\sqrt{3\,}}\,\phi _1\otimes\phi _2
  +\tfrac{1}{2\sqrt{3\,}}\,\phi _2\otimes\phi _1+\tfrac{1}{2\sqrt{3\,}}\,\phi _2\otimes\phi _2
\end{equation*}
The corresponding matrix becomes
\begin{equation*}
C=\frac{1}{2\sqrt{3\,}}\begin{bmatrix}3&1\\1&1\\\end{bmatrix}
\end{equation*}
Hence,
\begin{equation*}
\ab{C}^2=\frac{1}{6}\begin{bmatrix}5&2\\2&1\\\end{bmatrix},\quad\ab{C}^4=\frac{1}{36}\begin{bmatrix}29&12\\12&5\\\end{bmatrix}
\end{equation*}
We conclude that $\rmtr\paren{\ab{C}^4}=17/18$ and hence,
\begin{equation*}
e(\psi )=\paren{1-\frac{17}{18}}^{1/2}=\frac{1}{3\sqrt{2\,}}\rlap{$\qquad\qquad \Box$}
\end{equation*}
\end{exam}

\end{document}